\documentclass[letter,11pt]{article}

\usepackage[english]{babel}
\usepackage[utf8]{inputenc} 
\usepackage[T1]{fontenc}    
\usepackage{lmodern}
\usepackage{url}            
\usepackage{booktabs}       
\usepackage{amsfonts}       
\usepackage{nicefrac}       
\usepackage{microtype}      
\usepackage{enumitem}
\usepackage[colorlinks,allcolors=blue]{hyperref}

\usepackage{dsfont}
\usepackage{subfig}
\usepackage{wrapfig, lipsum}
\usepackage{graphicx}
\usepackage[titletoc,title]{appendix}
\usepackage{amsmath,amsthm,amssymb,mathtools}
\usepackage{algorithm}

\usepackage{float}
\usepackage{braket}

\usepackage{amsmath}
\usepackage{amssymb}
\usepackage{amsthm}
\usepackage{graphicx}
\usepackage{fullpage}

\theoremstyle{plain}

\newtheorem{theorem}{Theorem}[section]
\newtheorem*{theorem*}{Theorem}

\newtheorem{proposition}[theorem]{Proposition}
\newtheorem*{proposition*}{Proposition}
\newtheorem{lemma}[theorem]{Lemma}
\newtheorem*{lemma*}{Lemma}
\newtheorem{corollary}[theorem]{Corollary}
\newtheorem*{conjecture*}{Conjecture}

\newtheorem*{fact*}{Fact}

\newtheorem*{hypothesis*}{Hypothesis}

\newtheorem{claim}[theorem]{Claim}
\newtheorem*{claim*}{Claim}

\theoremstyle{definition}
\newtheorem{definition}[theorem]{Definition}

\theoremstyle{remark}

\newtheorem*{remark*}{Remark}

\newtheorem*{observation*}{Observation}

\newcommand{\R}{\mathbb{R}}

\newcommand{\poly}{\mathrm{poly}}

\newcommand{\norm}[1]{\lVert #1 \rVert}

\newcommand{\iprod}[1]{\langle#1\rangle}

\newcommand{\Esymb}{\mathbb{E}}
\newcommand{\Psymb}{\mathbb{P}}

\DeclareMathOperator*{\E}{\Esymb}

\DeclareMathOperator*{\ProbOp}{\Psymb}

\renewcommand{\Pr}{\ProbOp}

\renewcommand{\epsilon}{\varepsilon}

\newcommand{\Ocal}{\mathcal{O}}
\newcommand{\tOcal}{\widetilde{\mathcal{O}}}

\title{Efficient Tensor Decomposition\thanks{Chapter 19 of the book {\em Beyond the Worst-Case Analysis of Algorithms} ~\cite{bwcabook}.}}

\author{Aravindan Vijayaraghavan\thanks{
  Department of Computer Science,
  Northwestern University. Supported by the National Science Foundation (NSF) under Grant No.~CCF-1652491, CCF-1637585 and CCF-1934931. {\tt aravindv@northwestern.edu}. } }

\date{}

\begin{document}

\maketitle

\begin{abstract}
This chapter studies the problem of decomposing a tensor into a sum of constituent rank one tensors. While tensor decompositions are very useful in designing learning algorithms and data analysis, they are NP-hard in the worst-case. We will see how to design efficient algorithms with provable guarantees under mild assumptions, and using beyond worst-case frameworks like smoothed analysis. 

\end{abstract}

\tableofcontents

\section{Introduction To Tensors}

Tensors are multi-dimensional arrays, and constitute natural
generalizations of matrices. Tensors are fundamental linear algebraic entities, and widely used in physics, scientific computing and signal processing to represent multi-dimensional data or capture multi-wise correlations.  The different dimensions of the array are called the {\em modes} and the {\em order} of a tensor is the number of dimensions or modes of the array, as shown in Figure~\ref{fig:orders}. 
The order of a tensor also corresponds to the number of indices needed to specify an entry of a tensor. Hence every $(i_1,i_2, i_3) \in [n_1] \times [n_2] \times [n_3]$ specifies an entry of the tensor $T$ that is denoted by $T(i_1, i_2, i_3)$.
\begin{figure}[h]
  \begin{minipage}[c]{0.67\textwidth}
    \includegraphics[width=\textwidth]{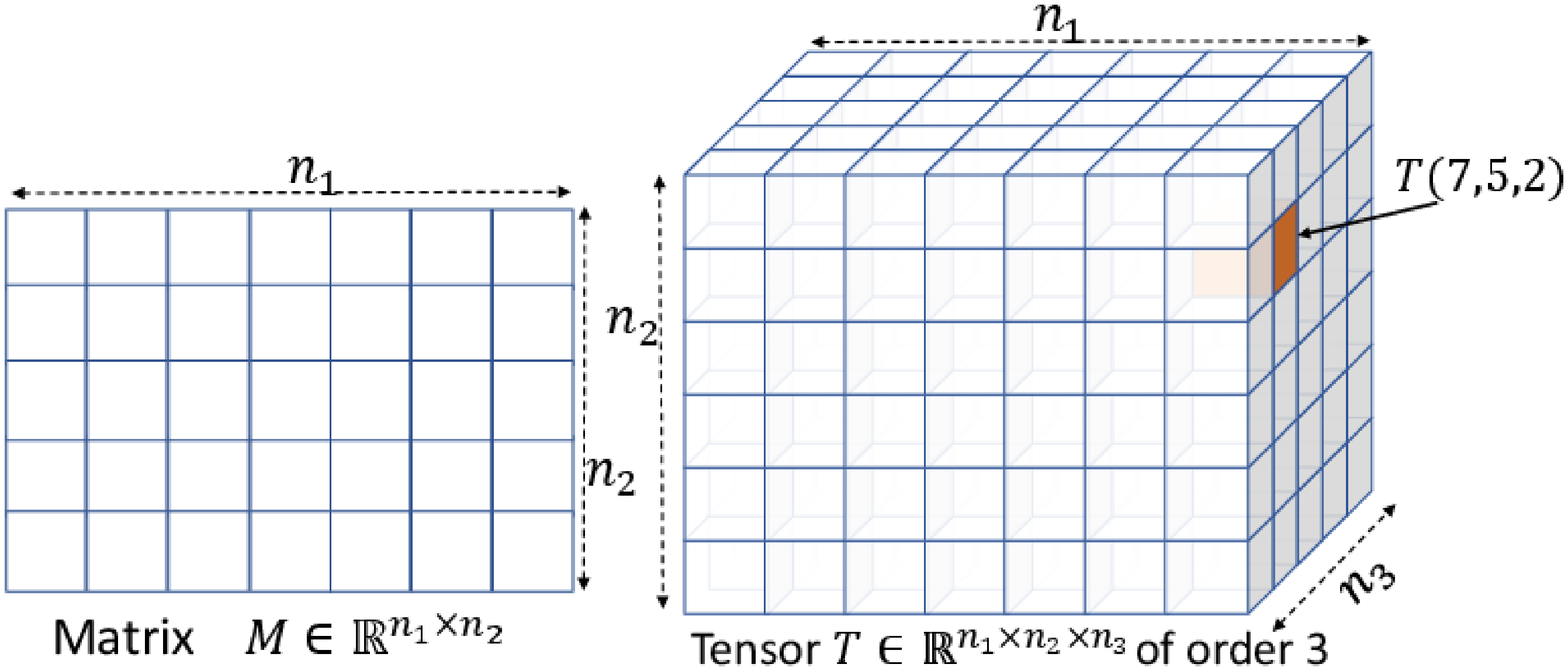}
  \end{minipage}\hfill
  \begin{minipage}[c]{0.3\textwidth}
    \caption{ shows a matrix $M$ which is a tensor of order $2$, and a tensor $T$ of order $3$ with $n_1=7,n_2=6,n_3=5$. The position of the entry $T(7,4,2)$ is highlighted. An order $1$ tensor corresponds to a vector, and an order $0$ tensor is a scalar. 
    } \label{fig:orders}
  \end{minipage}
\end{figure}
%
%
%

 While we have a powerful toolkit of algorithms like low rank approximations and eigenvalue decompositions for matrices, our algorithmic understanding in the tensor world is limited. As we will see soon many basic algorithmic problems like low-rank decompositions are NP-hard in the worst case for tensors (of order $3$ and above). But on the other hand, many higher order tensors satisfy powerful structural properties that are simply not satisfied by matrices. This makes them particularly useful for applications in machine learning and data analysis. In this chapter, we will see how we can indeed overcome this worst-case intractability under some natural non-degeneracy assumptions or using smoothed analysis, and 
also exploit these powerful properties for designing efficient learning algorithms.  


\subsection{Low-rank decompositions and rank}

We start with the definition of a rank one tensor. 
An order $\ell$ tensor $T \in \R^{n_1 \times \dots \times n_\ell}$ is {\em rank one} if and only if it can be written as an outer product $v_1 \otimes v_2 \otimes \dots \otimes v_\ell$ for some vectors $v_1 \in \R^{n_1}, \dots, v_\ell \in \R^{n_\ell}$ i.e., 
$$ T(i_1, i_2, \dots, i_\ell) = v_1(i_1) v_2(i_2) \dots v_\ell(i_\ell) ~~ \forall (i_1,\dots,i_\ell) \in [n_1] \times \dots \times [n_\ell].$$ 
Note than when $\ell=2$, this corresponds to being expressible as $v_1 v_2^T$. 

\begin{definition}\label{def:lowrank}[Rank $k$ decomposition]
A tensor $T$ is said to have a decomposition of {\em rank} $k$ iff it is expressible as the sum of $k$ rank one tensors i.e.,   
$$\exists \set{u^{(j)}_i | i \in [k], j \in [\ell]}, ~\text{s.t. }~~T=\sum_{i=1}^k u^{(1)}_i \otimes u^{(2)}_i \otimes \dots \otimes u^{(\ell)}_i.$$
Moreover $T$ has rank $k$ if and only if $k$ is the smallest natural number for which $T$ has a rank $k$ decompostion. 
\end{definition}
The vectors $\set{u^{(j)}_i: i \in [k], j \in [\ell]}$ are called the {\em factors} of the decomposition. To keep track of how the factors across different modes are grouped, we will use  $U^{(j)}=( u^{(j)}_i: i \in [k])$ for $j \in [\ell]$ to represent the factors. These ``factor matrices'' all have $k$ columns, one per term of the decomposition. Finally, we will also consider {\em symmetric tensors} -- a tensor $T$ of order $\ell$ is symmetric iff $T(i_1, i_2, \dots, i_r)=T(i_{\sigma(1)}, i_{\sigma(2)}, \dots, i_{\sigma(r)})$ for every permutation $\sigma$ over $\set{1,2,\dots,r}$ (see Exercise ~\ref{exer:symm_rank} for an exercise about decompositions of symmetric tensors). 
 
\paragraph{Differences from matrix algebra and pitfalls.} 
Observe that these definition of rank, low-rank decompositions specialize to the standard notions for matrices ($\ell=2$). However it is dangerous to use intuition we have developed from matrix algebra to reason about tensors because of several fundamental differences. 
Firstly, an equivalent definition for rank of a matrix is the dimension of the row space, or column space. This is {\em not true} for tensors of order $3$ and above. In fact for a tensor of order $\ell$ in $\R^{n^{\times \ell}}$, the rank as we defined it could be as large as $n^{\ell-1}$, while the dimension of the span of $n$ dimensional vectors along any of the modes can be at most $n$. The definition that we study in Definition~\ref{def:lowrank} (as opposed to other notions like Tucker decompositions) is motivated by its applications to statistics and machine learning. 



Secondly, much of the spectral theory for matrices involving eigenvectors and eigenvalues does not extend to tensors of higher order. For matrices, we know that the best rank-$k$ approximation consists of the leading $k$ terms of the SVD. 
However this is not the case for tensor decompositions. The best rank-$1$ approximation may not be a factor in the best rank-$2$ approximation.
Finally, and most importantly, the algorithmic problem of finding the best rank-$k$ approximation of a tensor is NP-hard in the worst-case, particularly for large $k$;\footnote{For small $k$, there are algorithms that find approximately optimal rank-$k$ approximations in time exponential in $k$ (see e.g.,  \cite{BCV,Woodrufftensor}).} for matrices, this is of course solved using singular value decompositions (SVD). In fact, this worst-case NP-hardness for higher order tensors is true for most tensor problems including computing the rank, computing the spectral norm etc.~\cite{Has90,HL}.

For all of the reasons listed above, and more,\footnote{
There are other definitional issues with the rank -- 
there are tensors of a certain rank, that can be arbitrarily well-approximated by tensors of much smaller rank i.e., the ``limit rank'' (or formally, the border rank) may not be equal to the rank of the tensor. See Exercise~\ref{exer:border_rank} for an example.} 
it is natural to ask, why bother with tensor decompositions at all? We will now see a striking property (uniqueness) satisfied by low-rank decompositions of most higher order tensors (but not satisfied by matrices), that also motivates many interesting uses of tensor decompositions.


\paragraph{Uniqueness of low-rank decompositions.}
A remarkable property of higher order tensors is that (under certain conditions that hold typically) their minimum
rank decompositions are unique upto trivial scaling and permutation.  This is in sharp contrast to matrix
decompositions. For any matrix $M$ with a rank $k \ge 2$ decomposition $M=U V^T=\sum_{i=1}^k u_i v_i^T$, there exists several other rank $k$ decompositions $M=U' (V')^T$, where $U'=UO$ and $V'=V O$ for any rotation matrix $O$ i.e., $O O^T=I_{k}$; in particular, the SVD is one of them. 
This {\em rotation problem}, is a common issue when using matrix decompositions in factor analysis (since we can only find the factors up to a rotation).  

The first uniqueness result for tensor decompositions was due to Harshman~\cite{Har70}(who in turn credits it to Jennrich), assuming what is known as the ``full rank condition''. 
In particular, if $T \in \R^{n \times n \times n}$ has a decomposition 
$$T=\sum_{i=1}^k u_i \otimes u_i \otimes u_i, \text{ s.t. }  \set{u_i: i \in [k]} \subset \R^n \text{ are linearly independent},$$ (or the factor matrix $U$ is full rank), then this is the {\em unique decomposition} of rank $k$ up to permuting the terms.  (The statement is actually a little more general and also handles non-symmetric tensors; see Theorem~\ref{thm:jennrich}).
Note that 
the full rank condition requires $k \le n$ (moreover it holds when the vectors are in general position in $n\ge k$ dimensions). 
What makes the above result even more surprising is that, the proof is algorithmic! We will in fact see the algorithm and proof in Section~\ref{sec:jennrich}. This will serve as the workhorse for most of the algorithmic results in this chapter.  
Kruskal~\cite{Kru77} gave a more general condition that guarantees uniqueness up to rank $3n/2 -1$, using a beautiful non-algorithmic proof. Uniqueness is also known to hold for {\em generic} tensors of rank $k=\Omega(n^2)$ (here ``generic'' means all except a measure zero set of rank $k$ tensors). 
We will now see how this remarkable property of uniqueness will be very useful for applications like learning latent variable models.   
  


\section{Applications to Learning Latent Variable Models} \label{sec:applications}

A common approach in unsupervised learning is to assume that the data (input) that is given to us is drawn from a probabilistic model with some latent variables and/or unknown parameters $\theta$, that is appropriate for the task at hand i.e., the structure we want to find. This includes mixture models like mixtures of Gaussians, topic models for document classification etc. A central learning problem is the efficient estimation of such latent model parameters from observed data. 

A necessary step towards efficient learning is
to show that the parameters are indeed identifiable after observing polynomially many samples. 
The method of moments approach, pioneered by Pearson, infers model parameters from empirical moments such as means, pairwise correlations and other higher order correlations. In general, very
high order moments may be needed for this approach to succeed and the unreliability of empirical estimates of these moments leads to large sample complexity (see e.g., \cite{MV10,BS10}). In fact, for latent variable models like mixtures of $k$ Gaussians, an exponential sample complexity of $\exp(\Omega(k))$ is necessary, if we make no additional assumptions. 

On the computational side, maximum likelihood estimation 
i.e., $\text{argmax}_\theta \Pr_\theta[data]$ 
is NP-hard for many latent variable models (see e.g., \cite{TD18}). Moreover iterative heuristics like expectation maximization (EM) tend to get stuck in local optima. Efficient tensor decompositions  when possible, present an algorithmic framework that is both statistically and computationally efficient, for recovering the parameters.  


\subsection{Method-of-moments via tensor decompositions: a general recipe}\label{sec:recipe}

The method-of-moments is the general approach of inferring parameters of a distribution, by computing empirical moments of the distribution and solving for the unknown parameters.   
The moments of a distribution over $\R^n$ are naturally represented by tensors. The covariance or the second moment is an $n \times n$ matrix, the third moment is represented by a tensor of order $3$ in $\R^{n \times n \times n}$ (the $(i_1,i_2,i_3)th$ entry is $\E[x_{i_2} x_{i_2} x_{i_3}]$), and in general the $\ell$th moment is a tensor of order $\ell$. More crucially for many latent variable models $\mathcal{D}(\bar{\theta})$ with parameters $\bar{\theta}$, the moment tensor or a suitable modification of it, has a low-rank decomposition (perhaps up to some small error) in terms of the unknown parameters $\bar{\theta}$ of the model. 
Low rank decompositions of the tensor can then be used to implement the general method-of-moments approach, with both statistical and computational implications. 
The uniqueness of the tensor decomposition then immediately implies {\em identifiability} of the model parameters (in particular, it implies a unique solution for the parameters)! Moreover, a computationally efficient algorithm for recovering the factors of the tensor gives an efficient algorithm for recovering the parameters $\bar{\theta}$. 

%
%
%
%

\paragraph{General Recipe.} This suggests the following algorithmic framework for parameter estimation.  Consider a latent variable model with model parameters $\bar{\theta}=(\theta_1, \theta_2, \dots, \theta_k)$. These could be one parameter each for the $k$ possible values of the latent variable (for example, in a mixture of $k$ Gaussians, the $\theta_i$ could represent the mean of the $i$th Gaussian component of unit variance). 

\begin{enumerate}
\item Define an appropriate statistic $\mathcal{T}$ of the distribution (typically based on moments) such that the expected value of $\mathcal{T}$ has a low-rank decomposition  
$$T=\E_{\mathcal{D}(\theta)}[\mathcal{T}]= \sum_{i=1}^k \lambda_i \theta_i^{\otimes \ell}, \text{ for some } \ell \in \mathbb{N}, \text{ and (known) scalars }\set{\lambda_i : i \in [k]}.$$ 
\item Obtain an estimate $\widetilde{T}$ of the tensor $T=\E[\mathcal{T}]$ from the data (e.g., from empirical moments) up to small error (denoted by the error tensor $E$). 

\item Use tensor decompositions to solve for the parameters $\bar{\theta}=(\theta_1, \dots, \theta_k)$ in the system $\sum_{i=1}^k \lambda_i \theta_i^{\otimes \ell} \approx \widetilde{T}$, to obtain estimates $\widehat{\theta}_1, \dots, \widehat{\theta}_k$ of the parameters. 
\end{enumerate}

The last step involving tensor decompositions is the technical workhorse of the above approach, both for showing identifiability, and getting efficient algorithms. Many of the existing algorithmic guarantees for tensor decompositions  (that hold under certain natural conditions about the decomposition e.g., Theorems~\ref{thm:jennrich} and \ref{thm:main}) provably {\em recover} the rank-$k$ decomposition, thereby giving algorithmic proofs of uniqueness as well. 
However, the first step of designing the right statistic $\mathcal{T}$ with a low-rank decomposition requires a lot of ingenuity and creativity. In Section~\ref{sec:casestudies} we will see two important latent variable models that will serve as our case studies. You will see another application in the next chapter on topic modeling. 

\paragraph{ Need for robustness to errors.}
So far, we have completely ignored sample complexity considerations by assuming access to the exact expectation $T=\E[\mathcal{T}]$, so the error $E=0$ (this requires infinite samples). In polynomial time, the algorithm can only access a polynomial number of samples. Estimating a simple 1D statistic up to $\epsilon=1/\poly(n)$ accuracy typically requires $\Omega(1/\epsilon^2)$ samples, the $\ell$th moment of a distribution requires $n^{O(\ell)}$ samples to estimate up to inverse polynomial error (in Frobenius norm, say). Hence, to obtain polynomial time guarantees for parameter estimation, it is vital for the tensor decomposition guarantees to be noise tolerant i.e., {\em robust} up to inverse polynomial error (this is even assuming no model mis-specification). Fortunately, such robust guarantees do exist -- in Section~\ref{sec:jennrich}, we will show robust analogue of Harshman's uniqueness theorem and related algorithms (see also \cite{BCV} for a robust version of Kruskal's uniqueness theorem). Obtaining robust analogues of known uniqueness and algorithmic results is quite non-trivial and open in many cases (see Section~\ref{sec:open}).  


\subsection{Case studies}\label{sec:casestudies}\label{sec:gmm}

\paragraph{Case Study 1: Mixtures of Spherical Gaussians.} 
Our first case study is mixtures of Gaussians. They are perhaps the most widely studied latent variable model in machine learning for clustering and modeling heterogenous populations.  
We are given random samples, where each sample point $x \in \R^n$ is drawn independently from one of $k$ Gaussian components according to mixing weights $w_1, w_2, \dots, w_k$, where each Gaussian component $j \in [k]$ has a mean $\mu_j \in \R^n$ and a 
covariance $\sigma^2_j I \in \R^{n \times n}$. 
The goal is to estimate the parameters $\set{(w_j, \mu_j, \sigma_j): j \in [k]}$ up to required accuracy $\epsilon>0$ in time and number of samples that is polynomial in $k,n, 1/\epsilon$. 
Existing algorithms based on method of moments have sample complexity and running time that is exponential in $k$ in general~\cite{MV10, BS10}. However, we will see that as long as certain non-degeneracy conditions are satisfied, tensor decompositions can be used to get tractable algorithms that only have a polynomial dependence on $k$ (in Theorem~\ref{thm:gmm:fullrank} and Corollary~\ref{corr:gmm:overcomplete}). 


For the sake of exposition, we will restrict our attention to the uniform case when the mixing weights are all equal and variances $\sigma_i^2 = 1, ~\forall i \in [k]$.
Most of these ideas also apply in the more general setting~\cite{HK13}.  

\newcommand{\mom}{\text{Mom}}
For the first step of the recipe, we will design a statistic that has a low-rank decomposition in terms of the means $\set{\mu_i : i \in [k]}$. 
\begin{proposition} \label{prop:gmm:moments}
For any integer $\ell \ge 1$, one can compute efficiently a statistic $\mathcal{T_\ell}$ from the first $\ell$ moments such that $\E[\mathcal{T}]= T_\ell := \sum_{i=1}^k \mu_i^{\otimes \ell}$.
\end{proposition}

Let $\eta \sim N(0,I)$ denote a Gaussian r.v. The expected value of the statistic $x^{\otimes \ell}$ 
\begin{align}
\mom_\ell:= \E [x^{\otimes \ell}] &= \sum_i w_i \E_\eta[(\mu_i+\eta)^{\otimes \ell}] = \frac{1}{k} \sum_{i=1}^k  \sum_{\substack{x_j \in \{\mu_i, \eta\} \\ \forall j \in [\ell]}} \E_\eta \big[ \bigotimes_{j=1}^\ell x_j \big]. \label{eq:iserlis1}
\end{align}

Now the first term in the inner expansion (where every $x_j= \mu_i$) is the one we are interested in, so we will try to ``subtract'' out the other terms using the first $(\ell-1)$ moments of the distribution. 
Let us consider the case when $\ell=3$ to gain more intuition. As odd moments of $\eta$ are zero, we have
\begin{align*}
\mom_3 := \E [x^{\otimes 3}] &= \frac{1}{k} \sum_{i=1}^k  \Big( \mu_i^{\otimes 3} + \E_\eta[\mu_i \otimes  \eta \otimes \eta] + \E_\eta[\eta \otimes  \eta \otimes \mu_i] + \E_\eta[\eta \otimes \mu_i \otimes \eta] \Big)\\
 &= T_3 + \Big(\mom_1 \otimes  I + \text{ two other known terms} \Big).
\end{align*}
Hence, we can obtain the required tensor $T_3$ using a combination of $\mom_3$ and $\mom_1$; the corresponding statistic is $x^{\otimes 3} - (x \otimes I + \text{two other known terms})$. 
We can use a similar inductive approach for obtaining $T_\ell$ 
(or use Iserlis identity that expresses higher moments of a Gaussian in terms of the mean and covariance)\footnote{An alternate trick to obtain a statistic $T_\ell$ that only loses constant factors in the dimension involves looking at an off-diagonal block of the tensor $\mom_\ell$ after partitioning the $n$ co-ordinates into $\ell$ equal sized blocks.}.  
\paragraph{Case study 2: Learning Hidden Markov Models (HMMs).}\label{sec:hmm}
Our next example is HMMs which are extensively used for data with a sequential structure. 
In an HMM, there 
is a hidden state sequence $Z_1,Z_2,\dots,Z_m$ taking values in $[k]$, that forms a stationary Markov chain $Z_1 \rightarrow Z_2 \rightarrow \dots \rightarrow Z_m$ with transition matrix $P$ and initial distribution $w=\{w_j\}_{j \in [k]}$ (assumed to be the stationary distribution). The observation $X_t$ 
is represented by a vector in $x^{(t)} \in \R^n$. Given the state $Z_t$ at time $t$, $X_t$ is conditionally independent of all other observations and states. The observation matrix is denoted by $\mathcal{O} \in \R^{n\times k}$; the columns of $\mathcal{O}$ represent the means of the observation $X_t \in \R^n$ conditioned on the hidden state $Z_t$ i.e., $\E[X_t \vert Z_t=i] = \mathcal{O}_i$, where $\mathcal{O}_i$ represents the $i$th column of $\mathcal{O}$. We also assume that $X_t$ satisfies strong enough concentration bounds to use empirical estimates. 
The parameters are $P, \Ocal, w$.

We now define appropriate statistics following ~\cite{AMR09}. 
Let $m = 2\ell + 1$ for some $\ell$ to be chosen later. The statistic $\mathcal{T}$ is $X_{2\ell+1} \otimes X_{2\ell} \otimes \dots \otimes X_1$. We can also view this $(2\ell+1)$ moment tensor as a 3-tensor of shape $n^\ell \times n \times n^\ell$. The first mode corresponds to $X_\ell  \otimes X_{\ell-1} \otimes \ldots \otimes X_1$, the second mode is $X_{\ell + 1}$ and the third mode is $X_{\ell + 2} \otimes X_{\ell+3} \otimes \ldots X_{2\ell + 1}$. Why does it have a low-rank decomposition? We can think of the hidden state $Z_{\ell + 1}$ as the latent variable which takes $k$ possible values. 

\begin{proposition}\label{prop:hmm:moments}
The above statistic $\mathcal{T}$ has a low-rank decomposition $\sum_{i=1}^k A_i \otimes B_i \otimes C_i$ with factor matrices $A \in \R^{n^{\ell} \times k}$, $B \in \R^{n \times k}$, and $C \in \R^{n^{\ell} \times k}$ s.t. $\forall i \in [k]$,
$$A_i=\E[ \otimes_{j=\ell}^{1} X_j | Z_{\ell + 1}=i],~~ B_i= \E[X_{\ell+1} | Z_{\ell+1}=i] ,\text{ and } C_i=\E[ \otimes_{j=\ell+2}^{2\ell+1}X_j | Z_{\ell+1}=i].$$
Moreover, $O,P$ and $w$ can be recovered from $A,B,C$. 
\end{proposition}
For $\ell=1$, $C=OP,B=O,A=OP'$ where $P'=\text{diag}(w) P^T \text{diag}(w)^{-1}$ is the reverse transition matrix. 
Tensor decompositions will allow for efficient recovery of $O,P,w$ in Theorem~\ref{thm:hmm:fullrank} and Section~\ref{sec:overcomplete:applications}.  
We leave the proof of Proposition~\ref{prop:hmm:moments} as Exercise~\ref{exer:hmm}. See \cite{AMR09} for more details.  

\section{Efficient Algorithms in the Full Rank Setting}
\subsection{Simultaneous Diagonization (Jennrich's algorithm)}
\label{sec:jennrich}

We now study Jennrich's algorithm (first described in \cite{Har70}), that gives theoretical guarantees for finding decompositions of third-order tensors under a natural non-degeneracy condition called the full-rank setting. 
Moreover this algorithm also has reasonable robustness properties, and can be used as a building block to handle more general settings and for many machine learning applications.
Consider a third-order tensor $T \in \R^{n \times m \times p}$ that has a decomposition of rank $k$:
$$  T=\sum_{i=1}^{k} u_i \otimes v_i \otimes w_i.$$
Our algorithmic goal is to recover the unknown factors $U,V,W$. 
Of course, we only hope to recover the factors up to some trivial scaling of vectors (within a rank-one term) and permuting terms. 
Note that our algorithmic goal here is much stronger than usual. 
This is possible because of uniqueness of tensor decompositions -- in fact, the proof of correctness of the algorithm also proves uniqueness!   

The algorithm considers two matrices $M_a, M_b$ that are formed by taking random linear combinations of the slices of the tensor as shown in Figure~\ref{fig:jennrich}. We will show later in \eqref{eq:jennrich:randcombs} that $M_a, M_b$ both have low-rank decompositions in terms of the unknown factors $\set{u_i, v_i}$. Hence, the algorithm reduces the problem of decomposing one third-order tensor into the problem of obtaining a ``simultaneous'' decomposition of the two matrices $M_a, M_b$ (this is also called simultaneous diagonalization).   


\begin{figure}[ht]
  \begin{minipage}[c]{0.35\textwidth}
    \includegraphics[width=\textwidth]{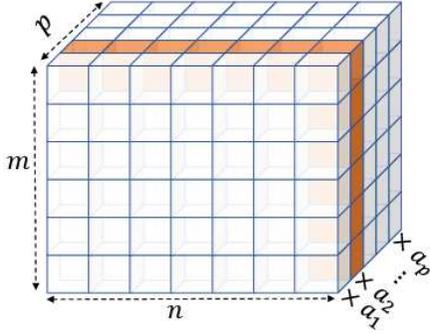}
  \end{minipage}\hfill
  \begin{minipage}[c]{0.65\textwidth}
    \caption{ shows a tensor $T$, and a particular matrix slice highlighted in orange (corresponding to $i_3=2$). The linear combination of the slices $T(\cdot,\cdot,a)$ takes a linear combination of these matrix slices weighted according to $a \in \R^p$. The algorithm considers two matrices $M_a=T(\cdot,\cdot,a), M_b=T(\cdot,\cdot,b)$ for two randomly chosen vectors $a,b \in \R^p$. 
    } \label{fig:jennrich}
  \end{minipage}
\end{figure}


In the following algorithm, $M^{\dagger}$ refers to the pseudoinverse or the Moore-Penrose inverse of $M$ (if a rank-$k$ matrix $M$ has a singular value decomposition $M=U \Sigma V^T$ where $\Sigma$ is a $k \times k$ diagonal matrix, then $M^{\dagger}=V \Sigma^{-1} U^T$).

\begin{algorithm}[!h]

\noindent {\bf Input}: Tensor $T \in \R^{n \times m \times p}$.
\begin{enumerate}
\item Draw $a, b \sim N(0,\tfrac{1}{p})^p \in \R^p$ independently. Set $M_a = T(\cdot, \cdot, a), M_b = T(\cdot, \cdot, b)$.
\item Set $\set{u_i: i \in [k]}$ to be the eigenvectors corresponding to the $k$ largest (in magnitude) eigenvalues of $M_a (M_b)^{\dagger}$. 
Similarly let $\set{v_i: i \in [k]}$ be the eigenvectors corresponding to the $k$ largest (in magnitude) eigenvalues of $\left( (M_b)^{\dagger} M_a\right)^T$. 
\item Pair up $u_i, v_i$ if their corresponding eigenvalues are reciprocals (approximately).
\item Solve the linear system $T = \sum_{i = 1}^k u_i \otimes v_i \otimes w_i$ for the vectors $w_i$. \label{step:linearsystem}
\item Return factor matrices $U \in \R^{n \times k},V \in \R^{m \times k},W \in \R^{p \times k}$.
\end{enumerate}

\caption{Jennrich's Algorithm}\label{alg:decomp}
\end{algorithm}

%
%
%
%
%
%
%

In what follows, $\norm{T}_F$ denotes the Frobenius norm of the tensor ($\norm{T}_F^2$ is the sum of the squares of all the entries), and the condition number $\kappa$ of matrix $U \in \R^{n \times k}$ is given by $\kappa(U)= \sigma_1(U) / \sigma_k(U)$, where $\sigma_1 \ge \sigma_2 \ge \dots \ge \sigma_k \ge 0$ are the singular values. 
The guarantees (in terms of the error tolerance) will be inverse polynomial in the condition number $\kappa$, which is finite only if the matrix has rank $k$ (full rank). 

 \begin{theorem} \label{thm:jennrich}
Suppose we are given tensor $\widetilde{T}=T + E \in \R^{m \times n \times p}$, where $T$ has a decomposition
$T = \sum_{i = 1}^k u_i \otimes v_i \otimes w_i$ satisfying the following conditions:
\begin{enumerate}
\item Matrices $U=(u_i: i \in [k]), V=(v_i: i \in [k])$ have condition number at most $\kappa$, 
\item For all $i \neq j$, $\|\frac{w_i}{\|w_i\|} - \frac{w_j}{\|w_j\|} \|_2 \geq \delta$.
\item Each entry of $E$ is bounded by $\| T \|_F \cdot   \epsilon /\mbox{poly}(\kappa, \max\set{n,m,p}, \tfrac{1}{\delta})$. 
\end{enumerate}
Then the Algorithm~\ref{alg:decomp} on input $\widetilde{T}$ runs in polynomial time and returns a decomposition $\set{(\widetilde{u}_i, \widetilde{v}_i, \widetilde{w}_i): i \in [k]}$ s.t. there is a permutation $\pi:[k] \to [k]$ with
$$ \forall i \in [k],~~ \|\widetilde{u}_i \otimes \widetilde{v}_i \otimes \widetilde{w}_i - u_{\pi(i)} \otimes v_{\pi(i)} \otimes w_{\pi(i)} \|_F  \le \epsilon \norm{T}_F.$$
\end{theorem} 
We start with a simple claim that leverages the randomness in the Gaussian linear combinations $a,b$ (in fact, this is the only step of the argument that uses the randomization). Let $D_a := \mbox{diag}(a^T w_1, a^T w_2, \dots, a^T w_k)$ and $D_b:=\mbox{diag}(b^T w_1, b^T w_2, \dots, b^T w_k)$.

\begin{lemma}\label{lem:separation}
With high probability over the randomness in $a,b$, the diagonal entries of $D_a D_b^{-1}$ are separated from each other, and from $0$, i.e., 
$$ \forall i \in [k]~~ \left| \frac{\iprod{w_i, a}}{\iprod{w_i,b}}\right| > \frac{1}{\text{poly}(p)} , \text{ and }  \forall i \ne j~~ \left| \frac{\iprod{w_i, a}}{\iprod{w_i,b}} - \frac{\iprod{w_j, a}}{\iprod{w_j,b}} \right| > \frac{1}{\text{poly}(p)}.$$
\end{lemma}

The proof just uses simple anti-concentration of Gaussians and a union bound.  We now proceed to the proof of Theorem~\ref{thm:jennrich}.

\begin{proof} [Proof of Theorem~\ref{thm:jennrich}]
We first prove that when $E=0$, the above algorithm recovers the decomposition exactly. The robust guarantees  when $E \ne 0$ uses perturbation bounds for eigenvalues and eigenvectors.

\vspace{5pt}
\noindent {\em No noise setting ($E=0$).} Recall that $T$ has a rank $k$ decomposition in terms of the factors $U, V, W$. Hence 
\begin{align}
M_a &= \sum_{i=1}^k \iprod{a, w_i} u_i v_i^T = U D_a V^T,~~\text{ and similarly } M_b = U D_b V^T.
\label{eq:jennrich:randcombs}
\end{align}
Moreover $U, V$ are full rank by assumption, and diagonal matrices $D_a, D_b$ have full column rank of $k$ with high probability (Lemma~\ref{lem:separation}). Hence
\begin{align*}
M_a (M_b)^{\dagger} = U D_a V^T (V^T)^{\dagger} D_b^{\dagger} U^{\dagger} = U D_a D_b^{\dagger} U^{\dagger}\\
\text{ and } M_a^T (M_b^T)^{\dagger}= V D_a D_b^{\dagger} V^{\dagger}.
\end{align*}

Moreover from Lemma~\ref{lem:separation}, the entries of $D_a D_b^{\dagger}$ are 
 distinct and non-zero with high probability.   
Hence the column vectors of $U$ are eigenvectors of $M_a (M_b)^{\dagger}$ with eigenvalues $( \iprod{w_i,a}/\iprod{w_i, b}: i \in [k])$. Similarly, the columns of $V$ are eigenvectors of $(M_b^{\dagger} M_a)^T$ with eigenvalues $( \iprod{w_i,b}/\iprod{w_i, a}: i \in [k])$. 
Hence, the eigendecompositions of $M_a M_b^{\dagger}$ and $(M_b^{\dagger} M_a)^T$ are unique (up to scaling of the eigenvectors) with the corresponding eigenvalues being reciprocals of each other.

Finally, once we know $\set{u_i, v_i: i \in [k]}$ (up to scaling), step~\ref{step:linearsystem} solves a linear system in the unknowns $\set{w_i: i \in [k]}$.  A simple claim shows that the corresponding co-efficient matrix given by $\set{u_i v_i^T: i \in [k]}$ has ``full'' rank i.e., rank of $k$. Hence the linear system has a unique solution $W$ and algorithm recovers the decomposition.  


\vspace{5pt} 
\noindent{\em Robust guarantees ($E$ is non-zero).}
When $E \ne 0$, we will need to analyze how much the eigenvectors of $M_1:=M_a M_b^\dagger$ can change, under the (worst-case) perturbation $E$. The proof uses perturbation bounds for eigenvectors of matrices (which are much more brittle than eigenvalues) to carry out this analysis. 
We now give a a high-level description of the approach, while pointing out a couple of subtle issues and difficulties. The primary issue comes from the fact that the matrix $M_1=M_a M_b^{\dagger}$ is not a symmetric matrix (for which one can use the Davis-Kahan theorem for singular vectors). In our case, while we know that $M_1$ is diagonalizable, there is no such guarantee about $M'_1= M_a M_b^{\dagger}+E'$, where $E'$ is the error matrix that arises at this step due to $E$. The key property that helps us here is Lemma~\ref{lem:separation}, which ensures that all of the non-zero eigenvalues of $M_1$ are separated. In this case, we know the matrix $M'_1$ is also diagonalizable using a standard strengthening of the Gershgorin disc theorem. 
One can then use the separation in the eigenvalues of $M_1$ to argue that the eigenvectors of $M_1, M'_1$ are close, using ideas from the perturbation theory of invariant subspaces (see Chapter 5 of\cite{StewartSun}). See also \cite{GVX, BCMV} for a self-contained proof of Theorem~\ref{thm:jennrich}. 
\end{proof}

\subsection{Implications in learning applications}\label{sec:applications:fullrank}

These efficient algorithms that (uniquely) recover the factors of a low-rank tensor decomposition give polynomial time guarantees for learning non-degenerate instances of several latent variable models using the general recipe given in Section~\ref{sec:recipe}. 
This approach has been used for several problems including but not limited to, parameter estimation of  hidden markov models, phylogeny models, mixtures of Gaussians, independent component analysis, topic models, mixed community membership models, ranking models, crowdsourcing models, and even certain neural networks (see ~\cite{AGHKT12,Moitrabook} for excellent expositions on this topic) . 

For illustration, we give the implications for our two case studies. 
For Gaussian mixtures, the $k$ means are assumed to be linear independent (hence $n \ge k$). We apply Theorem~\ref{thm:jennrich} to the $\ell=3$ order tensor obtained from Proposition~\ref{prop:gmm:moments}. 
\begin{theorem}
\cite{HK13} \label{thm:gmm:fullrank} Given samples from a mixture of $k$ spherical Gaussians, there is an algorithm that learns the parameters up to $\epsilon$ error in $\poly(n,1/\epsilon, 1/\sigma_{k}(M))$ time (and samples), where $M$ is the $n \times k$ matrix of means. 
\end{theorem}

For hidden markov models, we assume that the columns of the observation matrix $\Ocal$, and the transition matrix $P$ are linear independent (hence $n \ge k$). We apply Theorem~\ref{thm:jennrich} to the $\ell=3$ order tensor obtained from Proposition~\ref{prop:hmm:moments}. 

\begin{theorem}\cite{MR06,HKZ12}\label{thm:hmm:fullrank}
Given samples with $m = 3$ consecutive observations (corresponding to any fixed window of length $3$) from an HMM model as in Section~\ref{sec:hmm}, with $\sigma_k(\Ocal) \ge 1/\poly(n)$ and $\sigma_k(P) \ge 1/\poly(n)$, we can recover $P, \Ocal$ up to $\epsilon$ error in $\poly(n,1/\epsilon)$ time (and samples).
\end{theorem}


\section{Smoothed Analysis and the Overcomplete Setting} \label{sec:overcomplete}

The tensor decomposition algorithm we have seen in the previous section 
requires that the factor matrices have full column rank. 
As we have seen in Section~\ref{sec:applications:fullrank},  this gives polynomial time algorithms for learning a broad variety of latent variable models under the full-rank assumption. However, there are many applications in unsupervised learning where it is crucial that the hidden representation has much higher dimension (or number of factors $k$) than the dimension of the feature space $n$. Obtaining polynomial time guarantees for these problems using tensor decompositions requires polynomial time algorithmic guarantees when the rank is much larger than the dimension (in the full-rank setting $k \le n$, even when the $k$ factors are random or in general position in $\R^n$). 
Can we hope to obtain provable guarantees when the rank $k \gg n$?

This challenging setting when the rank is larger than the dimension is often referred to as the {\em overcomplete setting}. 
Tensor decompositions in the overcomplete setting is NP-hard in general. 
However for tensors of higher order, we will see in the rest of this section how Jennrich's algorithm can be adapted to get polynomial time guarantees even in very overcomplete settings for non-degenerate instances -- this will be formalized using smoothed analysis.  
 
 
\subsection{Smoothed analysis model.} \label{sec:smoothedmodel}

The smoothed analysis model for tensor decompositions models the situation when the factors in the decomposition are not worst-case.   

\begin{itemize}

\item An adversary chooses a tensor $T = \sum_{i=1}^k u^{(1)}_i \otimes u^{(2)}_i \otimes \dots \otimes u^{(\ell)}_i$.

\item Each vector $u^{(j)}_i$ is randomly ``$\rho$-perturbed'' using an independent Gaussian $N(0,\rho^2/n)^n$ with mean $0$ and variance $\rho^2/n$ in each direction \footnote{Many of the results in the section also hold for other forms of random perturbations, as long as the distribution satisfies a weak anti-concentration property, similar to the setting in Chapters 13-15; 
see \cite{ADMPSV18} for details.}.

\item Let $\widetilde{T} = \sum_{i=1}^k \tilde{u}^{(1)}_i \otimes \tilde{u}^{(2)}_i \otimes \dots \otimes \tilde{u}^{(\ell)}_i$. 
\item The input instance is $\hat{T}=\widetilde{T}+E$, where $E$ is some small potentially adversarial noise.
\end{itemize}

\noindent 
Our goal is to recover (approximately when $E \ne 0$) the $\ell$ sets of factors $U^{(1)}, \dots, U^{(\ell)}$ (up to rescaling and relabeling), where $U^{(j)}=(\widetilde{u}^{(j)}_i: i \in [k])$.  The parameter setting of interest is $\rho$ being at least some inverse polynomial in $n$, and the maximum entry of $E$ being smaller than some sufficiently small inverse polynomial $1/\poly(n,1/\rho)$. We will also assume that the Euclidean lengths of the factors $\set{u^{(j)}_i}$ is polynomially upper bounded. We remark that when $k \le n$ (as in the full-rank setting), Theorem~\ref{thm:jennrich} already gives smoothed polynomial time guarantees when $\varepsilon < \rho/\poly(n)$, since the condition number $\kappa \le \poly(n)/\rho$ with high probability.

{\em Remarks.} There is an alternate smoothed analysis model where the random perturbation is to each entry of the tensor itself, as opposed to randomly perturbing the factors of a decomposition. The two random perturbations are very different in flavor. When the whole tensor is randomly perturbed, we have $n^{\ell}$ ``bits'' of randomness, whereas when only the factors are perturbed we have $\ell n$ ``bits'' of randomness. 
On the other hand, 
the model where the whole tensor is randomly perturbed is unlikely to be easy from a computational standpoint, since this would likely imply randomized algorithms with good worst-case approximation guarantees. 

{Why do we study perturbations to the factors?} In most applications each factor represents a parameter e.g., a component mean in Gaussian mixture models. The intuition is that if these parameters of the model are not chosen in a worst-case configuration, we can potentially obtain vastly improved learning algorithms with such smoothed analysis guarantees. 

The smoothed analysis model can also be seen as the quantitative analog of ``genericity'' results that are inspired by results from algebraic geometry, particularly when we need robustness to noise. Results of this generic flavor give guarantees for all except a set of instances of zero measure.  
However, such results are far from being quantitative; as we will see later we typically need robustness to inverse polynomial error with high probability for polynomial time guarantees. 


\subsection{Adapting Jennrich's algorithm for overcomplete settings.}\label{sec:overcomplete:algo}

We will give an algorithm in the smoothed analysis setting for overcomplete tensor decompositions with polynomial time guarantees. In the following theorem, we consider the model in Section~\ref{sec:smoothedmodel} where the low-rank tensor $\tilde{T}=\sum_{i=1}^k \tilde{u}^{(1)}_i \otimes \tilde{u}^{(2)}_i \otimes \dots \otimes \tilde{u}^{(\ell)}_i$, and the factors $\set{\tilde{u}^{(j)}_i}$ are $\rho$-perturbations of the vectors $\set{u^{(j)}_i}$, which we will assume are bounded by some polynomial of $n$. The input tensor is $\tilde{T}+E$ where $E$ represents the adversarial noise.   
 
\begin{theorem}\label{thm:main}
Let $k \leq  n^{\lfloor \frac{\ell - 1}{2} \rfloor}/2$ for some constant  $\ell \in \mathbb{N}$, and $\varepsilon \in [0,1)$. There is an algorithm that takes as input a tensor $\hat{T}=\tilde{T}+E$ as described above, with every entry of $E$ being at most $\varepsilon/(n/\rho)^{O(\ell)}$ in magnitude, and runs in time $(n/\rho)^{O(\ell)}$ to recover all the rank one terms $\set{\otimes_{i=1}^\ell \tilde{u}^{(j)}_i : i \in [k]}$ up to an additive $\varepsilon$ error measured in Frobenius norm, with probability at least  $1 - \exp(- \Omega(n))$.   
\end{theorem}

To describe the main algorithmic idea, let us consider an order-$5$ tensor $T \in \R^{n \times n \times n \times n \times n}$. We can ``flatten" $T$ to get an order three tensor $$ T = \sum_{i = 1}^k \underbrace{u^{(1)}_i \otimes u^{(2)}_i}_{\mbox{factor}} \otimes \underbrace{u^{(3)}_i \otimes u^{(4)}_i}_{\mbox{factor}} \otimes \underbrace{u^{(5)}_i}_{\mbox{factor}}.$$ 
This gives us an order-$3$ tensor $T'$ of size $n^2 \times n^2 \times n$.
The effect of the ``flattening" operation on the factors can be described succinctly using the following operation. 

\begin{definition}[Khatri-Rao product]
The Khatri-Rao product of $A \in \R^{m \times k}$ and $B \in \R^{n \times k}$ is an $mn \times k$ matrix $U \odot V$ whose $i^{th}$ column is $u_i \otimes v_i$. 
\end{definition}

Our new order three tensor $T'$ 
also has a rank $k$ decomposition with factor matrices $U'=U^{(1)} \odot U^{(2)}, V'=U^{(3)} \odot U^{(4)}$ and $W'=U^{(5)}$ respectively. Note that the columns of $U'$ and $V'$ are in $n^2$ dimensions (in general they will be $n^{\lfloor (\ell-1)/2 \rfloor}$ dimensional). We could now hope that the assumptions on the condition number $U', V'$ in Theorem~\ref{thm:jennrich} are satisfied for $k = \omega(n)$. This is not true in the worst-case (see Exercise~\ref{exer:KRproduct} for the counterexample). 
However, we will prove this is true w.h.p. in the smoothed analysis model!


 
As the factors in $U^{(1)}, \dots, U^{(\ell)}$ are all polynomially upper bounded, the maximum singular value is also at most a polynomial in $n$. The following proposition shows high confidence lower bounds on the minimum singular value after taking the Khatri-Rao product of a subset of the factor matrices; this of course implies that the condition number has a polynomial upper bound with high probability. 

\begin{proposition}\label{prop:quantitativebound}
Let $\delta\in(0,1)$ be constants such that $k \le (1-\delta)n^{\ell}$. Given any $U^{(1)}, U^{(2)}, \dots, U^{(\ell)} \in \R^{n \times k}$, then for their random $\rho$-perturbations, 
we have 
\begin{align*}
    \mathbb{P}\Big[\sigma_k( \widetilde{U}^{(1)} \odot \widetilde{U}^{(2)} \odot \dots \odot \widetilde{U}^{(\ell)}) < \frac{c_1(\ell)\rho^\ell}{n^\ell}\Big]\leq k\exp\Big(-c_2(\ell)\delta n\Big).
\end{align*}
where $c_1(\ell), c_2(\ell)$ are constants that depend only on $\ell$.
\end{proposition}

The proposition implies that the conditions of Theorem~\ref{thm:jennrich} hold for the flattened order-$3$ tensor $T'$; in particular, the condition number of the factor matrices is now polynomially upper bounded with high probability. Hence by running Jennrich's algorithm to the order-$3$ tensor $T'$ recovers the rank-one factors w.h.p. as required in Theorem~\ref{thm:main}. The rest of the section outlines the proof of Proposition~\ref{prop:quantitativebound}.  


\vspace{5pt}
\noindent {\em Failure probability.}
We remark on a technical requirement about the failure probability (that is satisfied by the above proposition) for smoothed analysis guarantees. We need our bounds on the condition number or $\sigma_{\min}$ to  hold with a sufficiently small failure probability, say $n^{-\omega(1)}$ or even exponentially small (over the randomness in the perturbations). This is important because in smoothed analysis applications, the failure probability essentially describes the fraction of points around any given point that are {\em bad} for the algorithm. 
In many applications, the time/sample complexity has an inverse polynomial dependence on the least singular value. For example, if we have a guarantee that $\sigma_{\min} \ge \gamma$ with probability at least $ 1- \gamma^{1/2}$, then the probability of the running time exceeding $T$ (upon perturbation) is at most $ 1/\sqrt{T}$. Such a guarantee does not suffice to show that the expected running time is polynomial (also called polynomial smoothed complexity).


Note that our matrix $\widetilde{U}^{(1)} \odot \dots \odot \widetilde{U}^{(\ell)}$ is a random matrix with highly dependent entries e.g., there are only $kn \ell$ independent variables but $k n^{\ell}$ matrix entries.  
This presents very different challenges compared to well-studied settings in random matrix theory, where every entry is independent. 

While the least singular value can be hard to handle directly, it is closely related to the {\em leave-one-out distance}, which is often much easier to deal with.
\begin{definition}\label{def:leave-one-out}
Given a matrix $M \in \R^{n \times k}$ with columns $M_1, \ldots, M_k$, the leave-one-out distance of $M$ is
\begin{equation}
    \ell(M) = \min_{i \in [k]} ~\norm{\Pi_{-i}^{\perp} M_i}_2, ~\text{where } \Pi_{-i}^{\perp} \text{ is the projection matrix orthogonal to } \text{span}(\{M_j : j \neq i\}). \nonumber
\end{equation}
\end{definition}
The leave-one-out distance is closely related to the least singular value, up to a factor polynomial in the number of columns of $M$, by the following simple lemma. 
\begin{lemma}\label{lem:leave-one-out}
For any matrix $M \in \R^{n\times k}$, we have 
\begin{equation}
    \frac{\ell(M)}{\sqrt{k}} \le \sigma_{min}(M) \le \ell(M).
\end{equation}
\end{lemma}
 
\noindent The following (more general) core lemma that lower bounds the projection onto {\em any} given subspace of a randomly perturbed rank-one tensor implies Proposition~\ref{prop:quantitativebound}.  

\begin{lemma}\label{lem:quantitativebound}
Let $\ell \in \mathbb{N}$ and $\delta\in(0,\tfrac{1}{\ell})$ be constants, and let $W\subseteq \mathbb{R}^{n^{\times \ell}}$ be an arbitrary subspace of dimension at least $\delta n^\ell$. Given any $x_1,\cdots, x_\ell\in\mathbb{R}^n$, then their random $\rho$-perturbations $\tilde{x}_1,\cdots,\tilde{x}_\ell$ satisfy
\begin{align*}
    \Pr\Big[\norm{\Pi_W (\tilde{x}_1\otimes\tilde{x}_2\otimes\cdots\otimes\tilde{x}_\ell)}_2< \frac{c_1(\ell)\rho^\ell}{n^\ell}\Big]\leq \exp\Big(-c_2(\ell)\delta n\Big),
\end{align*}
where $c_1(\ell), c_2(\ell)$ are constants that depend only on $\ell$.
\end{lemma}
The polynomial of $n$ in the exponent of the failure probability is tight; however it is unclear what the right polynomial dependence of $n$ in the least singular value bound, and the right dependence on $\ell$ should be. 
The above lemma can be used to lower bound the least singular value of the matrix $\widetilde{U}^{(1)} \odot  \dots \odot \widetilde{U}^{(\ell)}$ in Proposition~\ref{prop:quantitativebound} as follows: we can lower bound the leave-one-out distance of Lemma~\ref{lem:leave-one-out} by applying Lemma~\ref{lem:quantitativebound} for each column $i \in [k]$ with $W$ being the subspace given by $\Pi^{\perp}_{-i}$ and $x_1,\dots, x_\ell$ being $u_i^{(1)}, \dots, u_i^{(\ell)}$; a union bound over the $k$ columns gives Proposition~\ref{prop:quantitativebound}.  
The first version of this lemma was proven in Bhaskara et al.~\cite{BCMV} with worse polynomial dependencies both in lower bound on the condition number, and in the exponent of the failure probability. The improved statement presented here and proof sketched in Section~\ref{sec:smoothedproof} are based on Anari et al.~\cite{ADMPSV18}.


\paragraph{Relation to anti-concentration of polynomials.}
We now briefly describe a connection to anti-concentration bounds for low-degree polynomials, and describe a proof strategy that yields a weaker version of Lemma~\ref{lem:quantitativebound}. Anti-concentration inequalities (e.g., the Carbery-Wright inequality) 
for a degree-$\ell$ polynomial $g : \R^n \rightarrow \R$ with $\norm{g}_2 \ge \eta$, and $x \sim N(0,1)^n$ are of the form 
\begin{equation}\label{eq:CWinequality}
\Pr_{x \sim N(0,1)^n}\Big[|g(x) - t| < \varepsilon \eta \Big]\leq O(\ell) \cdot \varepsilon^{1/\ell} . 
\end{equation}
  
This can be used to derive a weaker version of Lemma~\ref{lem:quantitativebound} with an inverse polynomial failure probability, by considering a polynomial whose co-efficients ``lie'' in the subspace $W$. As we discussed in the previous section, this failure probability does not suffice for expected polynomial running time (or polynomial smoothed complexity). On the other hand, Lemma~\ref{lem:quantitativebound} manages to get an inverse polynomial lower bound with exponentially small failure probability, by considering $n^{\Omega(1)}$ different polynomials. 
In fact one can flip this around and use Lemma~\ref{lem:quantitativebound} to show a vector-valued variant of the Carbery-Wright anti-concentration bounds, where if we have $m \ge \delta n^{\ell}$ ``sufficiently different'' polynomials $g_1, g_2, \dots, g_m: \R^n \to \R$ each of degree $\ell$, then we can get $\varepsilon^{c(\ell) \delta n}$ where $c(\ell)>0$ is a constant, for the bound in \eqref{eq:CWinequality}. The advantage is that while we lose in the ``small ball'' probability with the degree $\ell$, we gain an $\delta n$ factor in the exponent on account of having a vector valued function with $m$ co-ordinates. See \cite{BCPV} for a statement and proof.

\subsection{Proof Sketch of Lemma~\ref{lem:quantitativebound}} \label{sec:smoothedproof}

\newcommand{\tilx}{\tilde{x}}
\newcommand{\tily}{\tilde{y}}

The proof of Lemma~\ref{lem:quantitativebound} is a careful inductive proof. We sketch the proof for $\ell \le 2$ to give a flavor of the arguments involved. See \cite{ADMPSV18} for the complete proof. For convenience, let $\tilx:=\tilde{x}^{(1)}$ and $\tily:=\tilde{x}^{(2)}$. The high level outline is the following. We will show that there exist $n \times n$ matrices $M_1, M_2, \dots, M_r \in W$ of  bounded length measured in Frobenius norm (for general $\ell$ these would be order $\ell$ tensors of length at most $n^{\ell/2}$) which additionally satisfy certain ``orthogonality'' properties; here $r=\Omega_\ell(\delta n^{\ell})$.  We will use the orthogonality properties and the random perturbations to extract enough ``independence'' across $\{\iprod{M_i,(\tilx \otimes \tily)}: i \in [r]\}$; this will allow us to conclude that at least one of these $r$ inner products is at least $\rho /\sqrt{n}$ in magnitude with probability $\ge 1-\exp(-\Omega(\delta n))$.

\vspace{5pt}
\noindent {\em What orthogonality property do we want? }
\vspace{5pt}

\noindent {\bf Case $\ell=1$}. Let us start with $\ell=1$. In this case we have a subspace $W \subset \R^n$ of dimension at least $\delta n$. Here we could just choose the $r$ vectors $v_1, \dots, v_r \in \R^n$ to be an orthonormal basis for $W$, to conclude the lemma, since $\iprod{v_i,g}$ are independent. However, let's consider a slightly different construction where $v_1, \dots v_r$ are not orthonormal which will allow us to generalize to higher $\ell >1$.     

\begin{claim}[for $\ell=1$]\label{claim:l1}
There exists a set of $v_1, \dots, v_r \in W$, and a set of distinct indices $i_1, i_2, \dots, i_r \in [n]$, for $r =\dim(W)$ such that for all $j \in \set{1, 2, \dots, r}$:

~~~~~ (a)~ $\norm{v_j}_\infty \le 1$, ~~ (b)~ $|v_j(i_j)| = 1$, ~~(c) ~ $v_j(i_{j'})=0$ for all $j' <j$.
\end{claim} 

Hence, each of the vectors $v_j$ has a non-negligible component orthogonal to the span of $v_{j+1}, \dots, v_r$. This will give us sufficient independence across the random variables $\iprod{v_1, \tilx}, \dots, \iprod{v_r,\tilx}$. Consider the $r$ inner products in reverse order i.e.,  $\iprod{v_r, \tilx}, \iprod{v_{r-1}, \tilx}, \dots, \iprod{v_1, \tilx}$. Let $\tilx = x+z$ where $z \sim N(0,\rho^2/n)^n$. First $\iprod{v_r, \tilx}=\iprod{v_r,x}+\iprod{v_r,z}$, where $\iprod{v_r,z}$ is an independent Gaussian $N(0,\rho^2/n)$ due to the rotational invariance of Gaussians. Hence for some absolute constant $c>0$, from simple Gaussian anti-concentration $|\iprod{v_r,x}|< c \rho/\sqrt{n}$ with probability $1/2$. Now, let us analyze the event $\iprod{v_j,x}$ is small, after conditioning on the values of $\iprod{v_{j+1},\tilx}, \dots, \iprod{v_{r},\tilx}$. By construction, $|v_j(i_j)|=1$, whereas $v_{j+1}(i_j)=\dots=v_{r}(i_j)=0$. Hence
\begin{align*}
\Pr\Big[ |\iprod{v_j,\tilx}| < \frac{c\rho}{\sqrt{n}} ~\Big|~ \iprod{v_{j+1},\tilx}, \dots,\iprod{v_r, \tilx} \Big] & \le \text{sup}_{t \in \R} \Pr\Big[ | z(i_j) -t | < \frac{c\rho}{\sqrt{n}}\Big] \le \frac{1}{2}. \\
\text{ Hence } \Pr\Big[ \forall j \in [r], ~ |\iprod{v_j,\tilx}| <  \frac{c\rho}{\sqrt{n}} \Big]& \le \exp(-r), ~\text{ as required}.
\end{align*}

\begin{proof}[Proof of Claim~\ref{claim:l1}]
We will construct the vectors iteratively. For the first vector, pick any vector $v_1$ in $W$, and rescale it so that $\norm{v_1}_\infty=1$; let $i_1 \in [n]$ be an index where $|v_1(i_1)|=1$. For the second vector, consider the restricted subspace $\set{ x \in W: x(i_1)=0}$. This has dimension $\dim(W)-1$; so we can again pick an arbitrary vector in it and rescale it to get the necessary $v_2$. We can repeat this until we get $r=\dim(W)$ vectors (when the restricted subspace becomes empty).      
\end{proof}

\paragraph{Proof sketch for $\ell=2$.} We can use a similar argument to come up with an analogous set of matrices $M_1, \dots, M_r$ inductively. It will be convenient to identify each of these matrices $M_j$ with an (row,column) index pair $I_j =(i_{j}, i'_{j}) \in [n] \times [n]$. We will also have a total order among all of the index pairs as follows. We first have a ordering among all the valid row indices $R=\set{i_j : j \in [r]}$ (say $i_1 \prec i_2 \prec \dots \prec i_r$). Moreover, among all index pairs $R_{i^*}$ in the same row $i^*$ (i.e., $R_{i^*}:=\set{I_j=(i^*, i'_j)}$), we have a total ordering (note that it could be the case that $(2,4) \prec (2,7)$ and $(3,7) \prec (3,4)$, since the orderings for $i^*=2$ and $i^*=3$ could be different).   

\begin{claim}[for $\ell=2$]\label{claim:l2}
Given any subspace $W \subset \R^{n \times n}$ of dimension $\dim(W) \ge \delta n^2$, there exists $r$ many (row,column) index pairs $I_1 \prec I_2 \prec \dots \prec I_r$ as outlined above, and a set of associated matrices $M_1, M_2 \dots, M_r$ such that for all $j \in \set{1, 2, \dots, r}$:
~~~ (a)~ $\norm{M_j}_\infty \le 1$, ~~ (b)~ $|M_j(I_j)| = 1$, \\
(c) ~ $M_j(I_{j'})=0$ for all $j' <j$ and $M_j(i_1,i_2)=0$ for any $i_1 \prec i_j, i_2 \in [n]$ where $I_j=(i_j, i'_j)$.\\
Further there are at least $|R| = \Omega(\delta n)$ valid row indices, and each of these indices has $\Omega(\delta n)$ index pairs associated with it. 
\end{claim} 
The approach to proving the above claim is broadly similar to that of Claim~\ref{claim:l1}. The proof repeatedly treats the vectors in $W$ as vectors in $\R^{n^2}$ and applies Claim~\ref{claim:l1} to extract a valid row with $\Omega(\delta n)$ valid column indices, and iterates. We leave the formal proof as Exercise~\ref{exer:proof}. 


Once we have Claim~\ref{claim:l2}, the argument for Lemma~\ref{lem:quantitativebound} is as follows. Firstly, $\norm{M_j}_2 \le n$ since $\norm{M_j}_\infty \le 1$. Hence, we just need to show that there exists $j \in [r]$ s.t. $|\iprod{M_j, \tilx \otimes \tily}| \ge c\rho/n$ in magnitude with probability $\ge 1-\exp(-\Omega(\delta n))$. 
Consider the vectors $\set{M_1 \tily, M_2 \tily, \dots, M_r \tily} \subset \R^n$ obtained by applying just $\tily$. For each valid row $i^* \in R$, consider only the corresponding vectors with row index $i^*$ from $\set{M_j \tily: j \in [r]}$ and set $v_{i^*}$ to be the vector with the largest magnitude entry in coordinate $i^*$.
By our argument for $\ell=1$ we can see that with probability at least $1-\exp(-\Omega(\delta n))$, $|v_{i^*}(i^*)| > \tau:=c\rho/\sqrt{n}$, for some constant $c>0$. Now by scaling these vectors $\set{v_i: i \in [R]}$ by at most $1/\tau$ each, we see that they satisfy Claim~\ref{claim:l1}. Hence, using the argument for $\ell=1$ again, we get Lemma~\ref{lem:quantitativebound}. 
Extending this argument to higher $\ell>2$ is technical, and we skip the details. 

\subsection{Implications for applications} \label{sec:overcomplete:applications}

The smoothed polynomial time guarantees for overcomplete tensor decompositions in turn imply polynomial time smoothed analysis guarantees for several learning problems. In the smoothed analysis model for these parameter estimation problems, the unknown parameters $\theta$ of the model are randomly perturbed to give $\tilde{\theta}$, and samples are drawn from the model with parameters $\tilde{\theta}$. 

However, as we alluded to earlier, the corresponding tensor decomposition problems that arise, e.g., from Proposition~\ref{prop:gmm:moments} and \ref{prop:hmm:moments} do not always fit squarely in the smoothed analysis model in Section~\ref{sec:smoothedmodel}. For example, the random perturbations to the factors $\set{u^{(j)}_i: i \in [k], j \in [\ell]}$ may not all be independent. In learning mixtures of spherical Gaussians, the factors of the decomposition are $\tilde{\mu}_i^{\otimes \ell}$ for some appropriate $\ell>1$, where $\tilde{\mu}_i$ is the mean of the $i$th component.
In learning hidden Markov models (HMMs), each factor is a sum of appropriate monomials of the form $\tilde{a}_{i_1} \otimes \tilde{a}_{i_2} \otimes \dots \otimes \tilde{a}_{i_{\ell}}$, where $i_1 i_2 \dots i_{\ell}$ correspond to {\em length-$\ell$ paths} in a graph.

Fortunately the bounds in Proposition~\ref{prop:quantitativebound} can be used to derive similar high confidence lower bounds on the least singular value for random matrices that arise from such applications using {\em decoupling} inequalities. For example, one can prove such bounds (as in Proposition~\ref{prop:quantitativebound}) for the $k \times n^{\ell}$ matrix where the $i$th column is $\tilde{\mu}_i^{\otimes \ell}$ (as required for mixtures of spherical Gaussians). 
Such bounds also hold for other broad classes of random matrices that are useful for other applications like hidden markov models; see \cite{BCPV} for details. 

In the smoothed analysis model for mixtures of spherical Gaussians, the means $\set{\mu_i :i \in [k]}$ are randomly perturbed. The following corollary gives polynomial time smoothed analysis guarantees for estimating the means of a mixture of $k$ spherical Gaussians. See \cite{BCMV,Belkinetal} for details.
\begin{corollary}[Mixture of $k$ spherical Gaussians in $n \ge k^{\epsilon}$ dimensions] \label{corr:gmm:overcomplete} For any $\epsilon>0, \eta >0$, there is an algorithm that in the smoothed analysis setting learns the means of a mixture of $k$ spherical Gaussians in $n \ge k^{\epsilon}$ dimensions up to accuracy $\eta>0$ with running time and sample complexity $poly(n, 1/\eta, 1/\rho)^{O(1/\epsilon)}$ and succeeds with probability at least $1-\exp(-\Omega(n))$. 
\end{corollary}

In the smoothed analysis setting for hidden markov models (HMM), the model is generated using a randomly perturbed observation matrix $\tOcal$, obtained by adding independent Gaussian random vectors drawn from  $N(0,\rho^2/n)^n$ to each column of $\mathcal{O}$. These techniques also give similar smoothed analysis guarantees for learning HMMs in the overcomplete setting when $n \ge k^{\epsilon}$ dimensions (using $O(1/\epsilon)$ consecutive observations), and under sufficient sparsity of the transition matrix. See \cite{BCPV} for details.  
%
%
Smoothed analysis results have also been obtained for other problems like overcomplete ICA~\cite{GVX}, learning mixtures of general Gaussians~\cite{GHK},  
other algorithms for higher-order tensor decompositions~\cite{MSS, BCPV}, and recovering assemblies of neurons~\cite{ADMPSV18}.



\section{Other Algorithms for Tensor Decompositions}

The algorithm we have seen (based on simultaneous  diagonalization) has provable guarantees in the quite general smoothed analysis setting. However, there are other considerations like running time and noise tolerance, for which the algorithm is sub-optimal -- for example, iterative heuristics like alternating least-squares or alternating minimization are more popular in practice because of faster running times~\cite{KB09}. 
There are several other algorithmic approaches for tensor decompositions that work under different conditions on the input. The natural considerations are the generality of the assumptions, and the running time of the algorithm. The other important consideration is the robustness of the algorithm to noise or errors. I will briefly describe a selection of these algorithms, and comment along these axes. As we will discuss in the next section, the different algorithms are incomparable because of different strengths and weaknesses along these three axes. 


\vspace{5pt} 

\noindent {\em Tensor Power Method.} The tensor power method gives an alternate algorithm for symmetric tensors in the full-rank setting that is inspired by the matrix power method. The algorithm is designed for symmetric tensors $T \in \R^{n \times n \times n}$ with an {\em orthogonal decomposition} of rank $k \le n$ of the form $\sum_{i=1}^k \lambda_i v_i^{\otimes 3}$ where the vectors $v_1, \dots, v_k$ are orthonormal.
 Note that not all matrices need to have such an orthogonal decomposition. However in many learning applications (where we have access to the second moment matrix), 
one can use a trick called {\em whitening} to reduce to the orthogonal decomposition case by a simple basis transformation.

The main component of the tensor power method is an iterative algorithm to find one term in the decomposition 
that repeats the following power iteration update (after initializing randomly) until convergence
$ z \leftarrow \frac{T(\cdot, z,z)}{\norm{T(\cdot,z,z)}_2}$. Here the vector $T(\cdot, z,z)=u$ where $u(i)=\sum_{i_2,i_3} T(i,i_2,i_3) z_{i_2} z_{i_3}$.
The algorithm then removes this component and recurses on the remaining tensor. This method is also known to be robust to inverse polynomial noise, and is known to converge quickly after the whitening. See \cite{AGHKT12} for such guarantees. 


\vspace{5pt} 

\noindent {\em FOOBI algorithm and variants. } 
In a series of works, Cardoso and others~\cite{Cardoso, DLCC} devised an algorithm, popularly called the {\em FOOBI algorithm} for symmetric decompositions of overcomplete tensors of order $4$ and above.  At a technical level, the FOOBI algorithm finds rank-one tensors in a linear subspace, by designing a ``rank-1 detecting gadget''. Recently, the FOOBI algorithm and generalizations have been shown to be robust to inverse polynomial error in the smoothed analysis setting for order $2\ell$ tensors up to rank $k \le n^{\ell}$ (see ~\cite{MSS,HSS19} for order $4$ and \cite{BCPV} for higher even orders).  

\vspace{5pt} 

\noindent {\em Alternating Minimization and Iterative Algorithms. } 
Recently, Anandkumar et al.~\cite{AGJ15} analyzed popular iterative heuristics like alternating minimization 
for overcomplete tensors of order $3$ 
and gave some sufficient conditions for both local convergence and global convergence. 
Finally, a closely related non-convex problem is that of computing the ``spectral norm'' i.e., maximizing $\iprod{T,x^{\otimes \ell}}$ subject to $\norm{x}_2=1$; under certain conditions one can show that the global maximizers are exactly the underlying factors.  The optimization landscape of this problem for tensors 
has also been studied recently
(see \cite{GeMa}). But these results all mainly apply to the case when the factors of the decomposition are randomly chosen, which is much less general than the smoothed analysis setting.

\vspace{5pt} 

\noindent {\em Sum-of-squares algorithms. } 
The sum-of-squares hierarchy (SoS) or the Lasserre hierarchy is a powerful family of algorithms based on semidefinite-programming. 
Algorithms based on SoS typically consider a related polynomial optimization problem with polynomial inequalities. 
A key step in these arguments is to give a low-degree sum-of-squares proof of uniqueness; this is then ``algorithmicized'' using the SoS hierarchy. 
SoS-based algorithms are known to give guarantees that can go to overcomplete settings even for order $3$ tensors (when the factors are random), and are known to have higher noise tolerance. In particular, they can handle order-$3$ symmetric tensors of rank $k = \tilde{O}(n^{1.5})$, when the factors are drawn randomly from the unit sphere~\cite{MSS}. 
The SoS hierarchy also gives robust variants of the FOOBI algorithm, and get quasi-polynomial time guarantees under other incomparable conditions~\cite{MSS}. SoS based algorithms are too slow in practice because of large polynomial running times. Some recent works explore an interesting middle-ground; they design spectral algorithms that are inspired by these SoS hierarchies, but have faster running times (see e.g., \cite{HSSS16}).  
%
%
%

\section{Discussion and Open questions} \label{sec:open}

The different algorithms for tensor decompositions are incomparable because of different strengths and weaknesses.
 A major advantage of SoS-based algorithms is their significantly better noise tolerance; in some settings it can go up to constant error measured in spectral norm (of an appropriate matrix flattening), while other algorithms can get inverse polynomial error tolerance at best. This is important particularly in learning applications, since there is significant modeling errors in practice. 
However, many of these results mainly work in restrictive settings where the factors are random (or incoherent). 
On the other hand, the algorithms based on simultaneous decompositions and variants of the FOOBI algorithm work in the significantly more general smoothed analysis setting, but their error tolerance is much poorer. Finally, iterative heuristics like alternating minimization are the most popular in practice because of their significantly faster running times; however known theoretical guarantees are significantly worse than the other methods.   

Another direction where there is a large gap in our understanding is about conditions and limits for efficient recovery. This is particularly interesting under conditions that guarantee that the low-rank decomposition is (robustly) unique, as they imply learning guarantees. We list a few open questions in this space.

For the special case of $3$-tensors in $\R^{n \times n \times n}$, recall that Jennrich's algorithm needs the factors to be linearly independent, hence $k \le n$. On the other hand, Kruskal's uniqueness theorem (and its robust analogue) guarantees uniqueness even up to rank $3n/2-1$. 
Kruskal in fact gave a more general sufficient condition for uniqueness in terms of what is known as the {\em Kruskal rank} of a set of vectors~\cite{Kru77}. 
But there is no known algorithmic proof!

\noindent {\bf Open Problem.} {\em
Is there a (robust) algorithm for decomposing a $3$-tensor $T$ under the conditions of Kruskal's uniqueness theorem?}

We also do not know if there is any smoothed polynomial time algorithm that works for rank $(1+\epsilon)n$ for any constant $\epsilon>0$. Moreover, we know powerful statements using ideas from algebraic geometry that {\em generic} tensors of order $3$ have unique decompositions up to rank $n^2/3$~\cite{CO12}. However, these statements are not robust to even inverse polynomial error. Is there a robust analogue of this statement in a smoothed analysis setting? These questions are also interesting for order $\ell$ tensors. 
Most known algorithmic results for tensor decompositions also end up recovering the decomposition (thereby proving uniqueness).
However, even for order-$3$ tensors with random factors, there is a large gap between conditions that guarantee uniqueness vs conditions that ensure tractability.

\noindent {\bf Open Problem.} {\em
Is there a (robust) algorithm for decomposing a $3$-tensor $T$ with random factors for rank $k=\omega(n^{3/2})$?}


\section*{Acknowledgments}
I thank Aditya Bhaskara for many discussions related to the chapter, and Tim Roughgarden, Aidao Chen, Rong Ge and Paul Valiant for their comments on a preliminary draft of this chapter. 

\bibliographystyle{alpha}
\bibliography{chap19}



\section{Exercises}
\begin{enumerate}
\item The symmetric rank of a symmetric tensor $T$ is the smallest integer $r>0$ s.t., $T$ can be expressed as $T=\sum_{i=1}^r u_i^{\otimes \ell}$ for some $\set{u_i}_{i=1}^k$. Prove that for any symmetric tensor of order $\ell$, the symmetric rank is at most $2^{\ell} \ell!$ times the rank of the tensor\footnote{Comon's conjecture asks if for every symmetric tensor, the symmetric rank is equal to the rank. A counterexample was shown recently by Shitov. It is open what the best gap between these two ranks can be as a function of $\ell$.}. 
{\em Hint: For $\ell=2$, if $u_i \otimes v_i$ was a term in the decomposition, we can express $u_i \otimes v_i + v_i \otimes u_i = \tfrac{1}{2}(u_i + v_i)^{\otimes 2} - \tfrac{1}{2}(u_i - v_i)^{\otimes 2}$.}
\label{exer:symm_rank}

\item Let $u,v \in \R^n$ be two orthonormal vectors, and consider the tensor $A=u \otimes u \otimes v+v \otimes u \otimes u+u \otimes v \otimes u$. Prove that it has rank $3$. Also prove that it can be arbitrarily well approximated by a rank $2$ tensor. \\
{\em Hint: Try to express $A$ as a difference of two symmetric rank one tensors with large entries (Frobenius norm of $\Theta(m)$), so that the error term is $O(1/m)$.}
 \label{exer:border_rank}

\item Construct an example of a matrix $U$ such that the Kruskal-rank of $U \odot U$ is at most twice the Kruskal-rank of $U$. 
{\em Hint: Express the identity matrix as $\sum_i u_i u_i^T$ for two different orthonormal basis. }   
\label{exer:KRproduct}

\item Prove Proposition~\ref{prop:hmm:moments}. \label{exer:hmm}

\item Complete the proof of Claim~\ref{claim:l2} and hence Lemma~\ref{lem:quantitativebound} for $\ell=2$. 
\label{exer:proof}

\end{enumerate}

\end{document}